\documentclass[preprint]{revtex4-2}

\usepackage{svmacro}
\usepackage{comment}
\usepackage[T5]{fontenc}
\usepackage[utf8]{inputenc}
\usepackage{xcolor}

\def\be{\begin{equation}}
\def\ee{\end{equation}}

\begin{document}

\title{Quantum Telegraph Behavior Without Photons}

\author{Truong-Son P. Văn}
\email{truongson.vanp@gmail.com}
\affiliation{ Ho Chi Minh City, Vietnam}

\author{Daniel Maienshein}
\affiliation{Department of Mathematics, University of Pittsburgh, Pittsburgh, PA, 15260, USA}

\author{David W. Snoke}
\affiliation{Department of Physics and Astronomy, University of Pittsburgh, and Pittsburgh Quantum Institute, Pittsburgh, PA, 15260, USA}

\begin{abstract}
	We show that a simple model of non-Hermitian noise gives rise to the telegraph switching behavior seen in experiments with single qubits, without any reference to the existence of photons as corpuscles. This lends support to a continuous collapse interpretation of quantum mechanics, but can also be viewed as a model of continuous detection of a steady-state process in the incoherent limit. We show explicitly that such a system obeys the Born rule for particle counting statistics, even though no particle behavior has been invoked at any point in the calculation.
\end{abstract}

\maketitle

\section{Introduction}

It has been argued that in quantum mechanics one can remove the idea of photons as individual corpuscles, and treat photons simply as resonances of the underlying quantum fields. The derivation of the properties of photons (and other particles) as field resonances is well known, and supports this viewpoint (for a review, see Ref. \onlinecite{snokeQM}), but some experiments seem to naturally lend themselves to an interpretation in terms of individual corpuscles; namely, experiments with ``clicks'' or ``counts.'' One such experiment is when a wave with amplitude of $N=1$ is sent to multiple photon detectors, and only one of them clicks, even if they are spacelike separated. Another is when a continuous electromagnetic field impinges on a single atom or two-level system, and that system responds with ``telegraph noise;'' that is, the state of the system jumps suddenly and randomly between two states, in what can be seen as discrete photon absorption and emission events. Telegraph noise has been  observed in trapped ions~\cite{SauterBlattNeuhauserToschekQuantumJumps1988, BergquistHuletItanoWinelandObservationQuantum1986, NagourneySandbergDehmeltShelvedOptical1986}, cavity fields~\cite{GleyzesKuhrGuerlinBernuDelegliseBuskHoffBruneRaimondHarocheQuantumJumps2007}, molecules~\cite{BascheKummerBrauchleDirectSpectroscopic1995},
and superconducting artificial atoms~\cite{VijaySlichterSiddiqiObservationQuantum2011}, just to list a few. 

It is common to interpret the stochastic nature of clicks and counts in quantum mechanics as arising from measurement of single photons. In the standard model of measurement, the clicks or counts seen in such experiments are the result of an intrinsically random, non-unitary process, which cannot be accounted for by the Schr\"odinger equation or by any equations derived from it, such as Lindbladian dynamics. For example, in the Copenhagen interpretation, such a non-unitary process is instigated by human knowledge. However, two recent developments in the theory of measurement allow us to take a different perspective. The first is the large body of work on ``weak'' measurement (for a review, see, e.g., Ref. \onlinecite{jordan}), which shows that what we encounter as all-or-nothing measurements can be seen as the limiting behavior of a large number of weak, partial measurements. As shown in prior work \cite{spont3,VanJordanSnokeMeasurementTime2025}, the effect of weak measurements can be represented by a differential operator  that can be added to the Hamiltonian of the Schr\"odinger equation. 

A second body of work that has arisen in recent decades is spontaneous collapse theory, which has had several variants such as continuous spontaneous localization (CSL) \cite{diosi1,diosi2,gisin3,bassi1,bassi, bassi2}, the Ghirardi-Rimini-Weber (GRW) model \cite{grw,tumulka,heat3}, and nonlocal stochastic collapse \cite{gisin,gisin2}.
A recent variation \cite{spont,spont2,spont3,VanJordanSnokeMeasurementTime2025} takes a similar approach, but argues that no net violation of energy conservation is needed, if a simple non-Hermitian term is added to the Schr\"odinger equation. It has been shown \cite{spont3,VanJordanSnokeMeasurementTime2025} that the non-Hermitian different operator added to the Schr\"odinger equation in this approach is identical to that derived from weak measurement theory. Therefore, one can take all the mathematical results of weak measurement theory, but dismiss the measurement apparatus altogether, and simply posit this differential operator as a universal stochastic noise term, with a noise level proportional to actual physical random fluctuations in the system. 

It has already been shown \cite{VanJordanSnokeMeasurementTime2025} that the clicking of one detector out of the midst of many can be interpreted as follows: an input optical plane-wave puts multiple atoms into a superposition of both the upper and lower states of a two-level system, and then a stochastic process causes one of the atoms to end up in its upper state, while jamming all of the others into the lower state; this can happen suddenly, with a measurement time that is not strongly affected by the number of atoms used as detectors. In other words, the random ``click'' at one atom comes ultimately not from a pre-existing photon corpuscle, but from a random walk due to quantum noise acting nonlocally on many atoms with discrete excitation states; it is not necessary at all to have a photon as a corpuscle to get just one of the detectors to click.

In this paper we show that the same type of analysis can be applied to the second type of experiment mentioned above, namely telegraph noise.  A single two-level system (qubit) is subject to a continuous, steady-state, incoherent pump, but shows stochastic jumps between its upper and lower state, with statistics that track with the expected probability according to the Born measurement rule. No photons as corpuscles are invoked to obtain this behavior.

Prior work on telegraphing behavior in the context of weak measurement theory \cite{jordan, WisemanMilburnQuantumMeasurement2010, korotkov1999continuous, BauerBernardTilloyComputingRates2015}  
has assumed the interpretation that the ``measurements'' are, at the end of the day, gaining human knowledge about what photon corpuscles are doing. However,  the entire effect of the ``measurements'' is to introduce stochastic noise, which, as we show here, reproduces the  quantum Born rule for particle counting, namely, that the probability of a count is proportional to the square of the amplitude of the wave function at the detector.  In any two-state system, low-level continuous pumping with energy loss and noise will result in telegraph noise, jumping between its two states, which obeys the Born rule. Therefore it is appealing to simply posit the quantum noise term as fundamental.

As discussed in Section~\ref{sect.classsical}, the quantum mechanical results obtained in this case are similar to classical telegraphing systems, but have important differences.

\section{The two-level model with continuous pumping and decay}

The standard Bloch equations for a qubit are \cite{snokeQM}
\begin{eqnarray}
	\frac{\partial U_1}{\partial t} &=&   - \frac{U_1}{T_2} +\tilde\omega U_2  \nonumber \\
	\nonumber \\
	\frac{\partial U_2}{\partial t} &=&   - \frac{U_2}{T_2} -\tilde\omega U_1 - \omega_R U_3'
	\nonumber \\
	\nonumber \\
	\frac{\partial U_3}{\partial t} &=&  - \frac{U_3+1}{T_1} +\omega_R U_2,
    \label{Blochstandard}
\end{eqnarray}
where $\vec{U}$ is the standard Bloch vector in the rotating frame, $T_2$ is a decoherence time and $T_1$ is an intrinsic energy-loss time, $\tilde\omega$ is the detuning of the pump frequency from the qubit resonance frequency, and $\omega_R$ is the Rabi oscillation frequency, proportional to the intensity of the pump wave. 

In the case of an incoherent pump, a term can be added to the last equation (see Appendix~\ref{blochder}) that acts oppositely to a $T_1$ loss process, namely $G(1-U_3)$, where $G$ is the incoherent pump rate, which is either zero or positive, and the $(1-U_3)$ factor accounts for Pauli exclusion, preventing the upper state from having occupation greater than 1. 
If we assume fully incoherent pumping, this gives us
\begin{eqnarray}
	\frac{\partial U_3}{\partial t} &=&  - \frac{U_3+1}{T} +G(1-U_3),
	\label{UBloch}
\end{eqnarray}
where we write $T_1 =T$ as the characteristic decay time.
In steady state, this has the solution
\begin{equation}
	U_3 = \frac{GT-1}{GT+1}.
\end{equation}
Since $U_3 = \langle N_e\rangle -  \langle N_g\rangle$, where $N_e$ is the  number in the excited state and $N_g$ is the number in the ground state, this can be interpreted as giving the probability of the qubit being in one of the two states; since $N_e + N_g =1$, we can also write $U_3 = 2N_e-1$. If we are looking at a single qubit over a long time, we can take these occupation numbers as a time average over jumps between the states. We expect that the ratio of time spent in the upper versus lower state will then be
\begin{eqnarray}
	\frac{\langle t_e\rangle}{\langle t_g\rangle} = \frac{\langle N_e\rangle}{\langle N_g\rangle} = \frac{\langle N_e\rangle}{1-\langle N_e\rangle}& =& \frac{(U_3+1)/2}{1-(U_3+1)/2} \nonumber\\
	&=& \frac{1+U_3}{1-U_3} \nonumber\\
	&=& \frac{1+ \displaystyle \frac{GT-1}{GT+1}}{1-\displaystyle \frac{ GT-1}{GT+1}} = GT.
	\label{avgstat}
\end{eqnarray}
In Appendix~\ref{sect.stat} we justify this prediction by
analyzing the behavior of $U_3$.

Therefore we expect in a telegraph scenario,  the fraction of time spent in the upper state will be proportional to the pumping rate.  However, nothing in the Bloch equations will give stochastic jumping. To obtain this, we add in the same term used in prior work,
\begin{eqnarray}
	\frac{\partial U_3^{\rm stoch}}{\partial t} &=& \epsilon (1-U_3^2),
	\label{Ustoch}
\end{eqnarray}
where $\epsilon$ is a small, fluctuating term that gives, if acting alone, a Martingale random walk of $U_3$. As discussed in the Introduction, this term can either be justified as a fundamental postulate of spontaneous collapse theory, or as the differential limiting behavior of continuous weak measurement.
The random walk generated by this term has been shown \cite{spont2} to give trapping (``collapse'') at the values $U_3 = \pm 1$ for states that begin in a superposition between these two values, with probability in agreement with the standard Born rule for quantum measurements.

If this stochastic term acts alone on a qubit, the random walk of the Bloch vector will end up at the top or the bottom of the sphere and stay there. However, if we add together the two terms (\ref{UBloch}) and (\ref{Ustoch}), we will have dynamical instability. It can easily be seen that the term (\ref{Ustoch}) never leads the Bloch vector to {\em exactly} $U_3 =1$ or $U_3 = -1$ (as in the quantum Zeno effect \cite{jordan, korotkov1999continuous,WisemanMilburnQuantumMeasurement2010}; the step size of the random walk decreases as $U_3$ approaches one of these end points, so that it becomes quite close, and trapped, so to speak, but never quite at the end. In this context, the ``force'' (\ref{UBloch}) has the effect of pushing the Bloch vector away from the ends: at $U_3 = 1$,
${\partial U_3}/{\partial t} =  - 2/T$, and at $U_3 = -1$,
${\partial U_3}/{\partial t} =  2G$. Therefore there is a nonzero chance that the Bloch vector will escape an end point, and undergo a new random walk in the middle of the sphere. As shown in Appendix~\ref{sect.stat}, when this occurs, the statistical prediction (\ref{avgstat}) can justified mathematically.

This model is essentially the same as that considered in the context of weak measurement theory with coherent-wave pumping (see, e.g., Ref.~\cite{jordan}, Chapters 5 and 6, and references therein), but with an incoherent pump. By casting this model only in terms of the non-unitary term (\ref{Ustoch}), we show that the full apparatus of weak measurement theory is not neeeded, however; simply positing a universal noise term of the form used here converges to the predictions of weak measurement theory in all known cases. That is, the term (\ref{Ustoch}) is by itself {\em sufficient} for particle-counting clicks that follow the Born rule of counting statistics.

\section{Numerical results}

As in previous work~\cite{VanJordanSnokeMeasurementTime2025},
we replace $\eps$ by $\alpha X_n/(\Delta t)^{1/2}$, where $\set{X_n}_{n\in \N}$ are
independent identically distributed (iid) random variables and $\alpha$ gives the ratio of the time scale of the incoherent pumping (which nominally requires a time $T/(GT+1)$ to traverse the Bloch sphere), and the time scale for the random walk to traverse the Bloch sphere. In the limit of $\Delta t \to 0$, this
satisfies
\begin{equation}
	\label{eq:SDE}
	d U_{3} = \left(-\frac{U_3 + 1}{T} + G(1 - U_3) \right)  dt + \alpha (1 - U_{3}^2) dW_t \,,
\end{equation}

\noindent where $W_t$ is one-dimensional Brownian motion.
We numerically simulated this equation via Euler-Maruyama~\cite{Higham2021} method with initial value $U_3(0) = 1$ using the Julia Language~\cite{Julia-2017}.

\begin{figure}[h!]
	\begin{center}
		\includegraphics[width=0.7\textwidth]{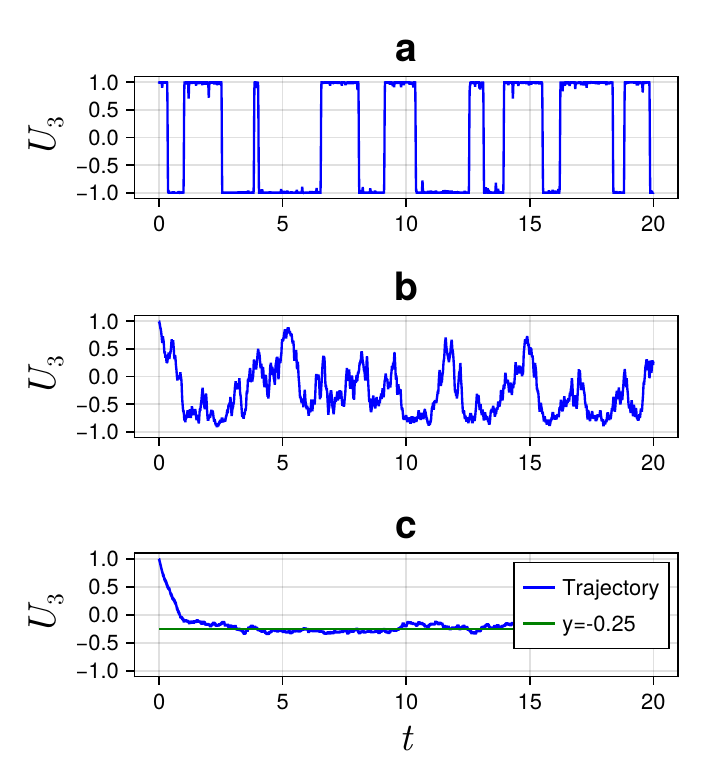}
	\end{center}
	\caption{a) A typical ``telegraph'' trajectory of~\eqref{eq:SDE}, for weak continuous driving and decay. Parameter values were $dt = 10^{-4}$,  $\alpha = 10$, $T = 1$, and $G = 0.6$, corresponding to $(GT-1)/(GT+1) = -0.25$.
		b) A typical noisy trajectory of~\eqref{eq:SDE} when the random fluctuation rate is high, for the same parameters as (1a) except $\alpha = 0.5$.
		c) A typical trajectory of~\eqref{eq:SDE} when the random fluctuation rate is low, for the same parameters as (1a) except $\alpha = 0.05$.
		In all cases the initial condition was $U_3(0) = 1$, and the data were smoothed by an instrumental resolution function with temporal width of $5\times 10^{-3}$.}
	\label{fig:one-path-unfiltered}
\end{figure}

The competing driving plus stochastic terms lead, as expected, to telegraph noise when simulated numerically.
Figure~\ref{fig:one-path-unfiltered}(a) shows a typical numerical simulation of $U_3$ as a function of time.
As seen in this figure, the system remains for the great majority of time in one state or the other, and only briefly spends time in a superposition of the two. Figure~\ref{fig:ratio} shows the ratio of time spent in the upper state to the time spent in the lower state, as the term $GT$ is varied. As seen in this figure, the prediction of (\ref{avgstat}) is confirmed to high accuracy.

\begin{figure}[h!]
	\begin{center}
		\includegraphics[width=0.6\textwidth]{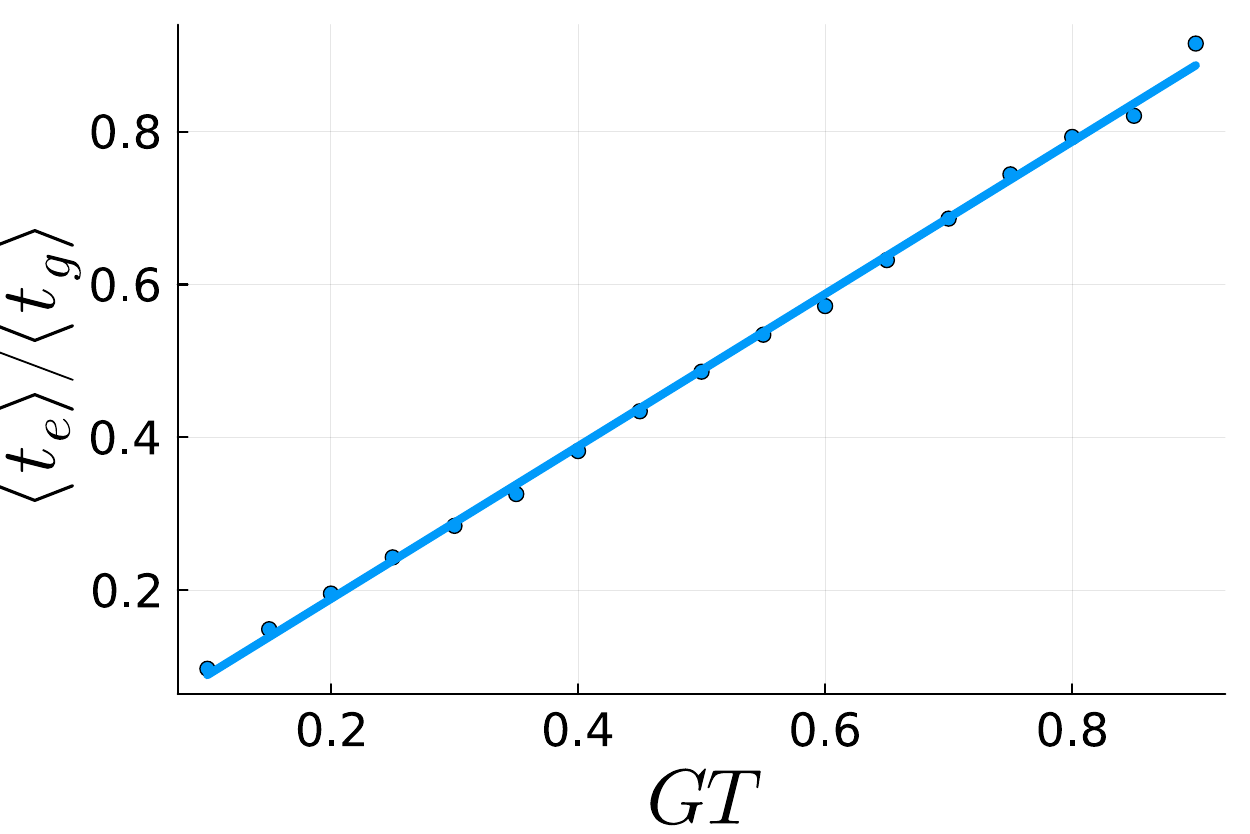}
	\end{center}
	\caption{Ratio of time spent in the upper versus lower states as $GT$ varies, averaged over $2\times 10^8$ time steps of each simulation. Presence in one of the two states was defined as $U_3$ within 0.05 of either $+1$ or $-1$. Parameter values were $dt = 10^{-4}$, $\alpha = 10$, and $T =1$. }\label{fig:ratio}
\end{figure}

\begin{figure}[h!]
	\begin{center}
		\includegraphics[width=0.6\textwidth]{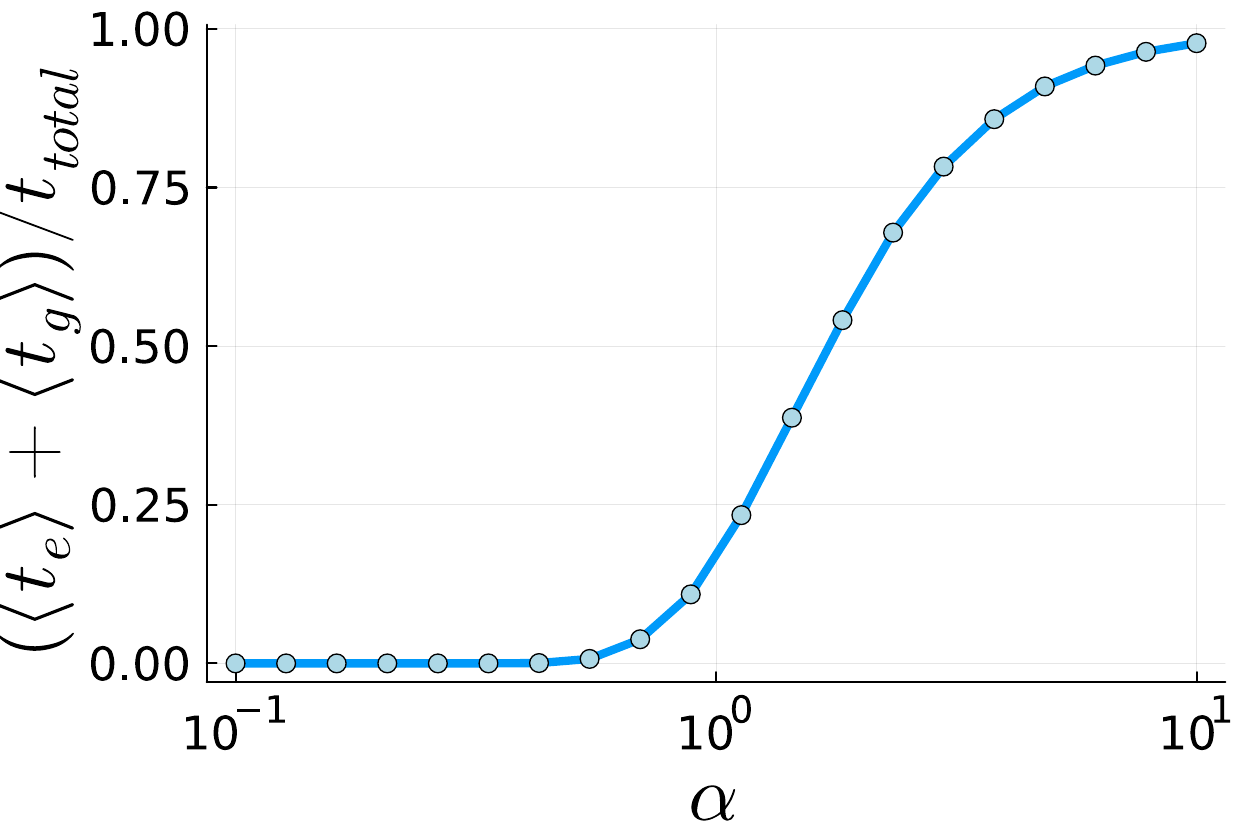}
	\end{center}
	\caption{Ratio of time spent in either the upper or lower state to the total evolution time, defined as in Figure~\ref{fig:ratio}, averaged over $10^8$ time steps of one simulation, while varying relative strength between incoherent pumping and the random walk. Parameter values were $GT = 0.6$, $dt = 10^{-4}$, and $T =1$.}\label{fig:varying-alpha}
\end{figure}

This clear telegraph behavior occurs when the collapse time, or measurement time, is short compared to the time scale for the driving term $GT$ to act. If the deterministic driving term is stronger, then we expect that the system will be forced away from an end point rapidly. This physically corresponds to the case when $GT$ is kept constant, but $G$ is increased (stronger electromagnetic field is used) and $T$ is decreased (faster decay); in photon language, the system is continuously rapidly absorbing and emitting photons. Figures~\ref{fig:one-path-unfiltered}(b) and (c) show example trajectories as $\alpha$ is reduced; in the limit of strong driving, the system settles to the steady-state value (\ref{avgstat}). Figure~\ref{fig:varying-alpha} shows how the telegraph behavior fades away as the ratio of the driving force to the measurement time increases: for driving weak compared to the noise (high $\alpha$), the system is almost always in one state or the other, while for strong driving, it spends much more time in random motion in a superposition in the middle of the two states.

\section{Comparison to other mode switching systems}
\label{sect.classsical}

\subsection{Weak measurement}

Telegraph behavior has been studied in the weak measurement literature, as well as the
quantum Zeno effect~\cite{jordan, korotkov1999continuous,WisemanMilburnQuantumMeasurement2010, PercivalQuantumState1998}.
A potential source for switching behavior is from
the randomness introduced to both the decoherence and incoherence by the measurement process,
via the following set of equations~\cite{jordan}:
\begin{align*}
	\frac{dx}{dt} & = -\frac{x}{2\tau_m} - \frac{xz}{\sqrt{\tau_m}} \, \xi - \epsilon y,            \\
	\frac{dy}{dt} & = -\frac{y}{2\tau_m} - \frac{yz}{\sqrt{\tau_m}} \, \xi + \epsilon x - \Delta z, \\
	\frac{dz}{dt} & = \frac{1 - z^2}{\sqrt{\tau_m}} \, \xi + \Delta y.
\end{align*}
Here, $\Delta$ is the tunneling rate, $\xi$ is a Gaussian process, $\tau_m$ is the characteristic meaurement time.
When there is no tunneling, the third equation for $z$ is exactly that of~\eqref{eq:SDE} without the pumping.
With this system, quantum Zeno effect was studied by varying $\tau_m$, which represents the noise introduced to the system by measurements.

In our system, we do not need the coherence to explicitly evolve randomly
and can still study the effect of noise to give the quantum Zeno effect by varying $\alpha$ (see Fig.~\ref{fig:one-path-unfiltered}).
We therefore see that in order to get telegraph jumps that give the Born statistics of quantum particles, it is sufficient to assert a deterministic pumping term as used in (\ref{UBloch}) and a noise term of the form (\ref{Ustoch}).

\subsection{Classical switching system based on potential wells}

Figure~\ref{QMexpt}(a) shows data from an experiment on quantum telegraph behavior, namely a trapped ion under continuous illumination and observation. As seen in this figure, our model simulates well the telegraph results. Similar behavior has been seen in the coupling of a single two-level system to an external macroscopic measurement system in the presence of continuous pumping
(see Refs.~\cite{JiangYaoSiyuChenWenluShiWilsonHoQuantumStochastic2025, HanzeMcMurtrieBaumannMalavoltiCoppersmithLothQuantumStochastic2021} and references therein).

\begin{figure}[h!]
	\begin{center}
		\begin{center}
			\includegraphics[width=0.6\textwidth]{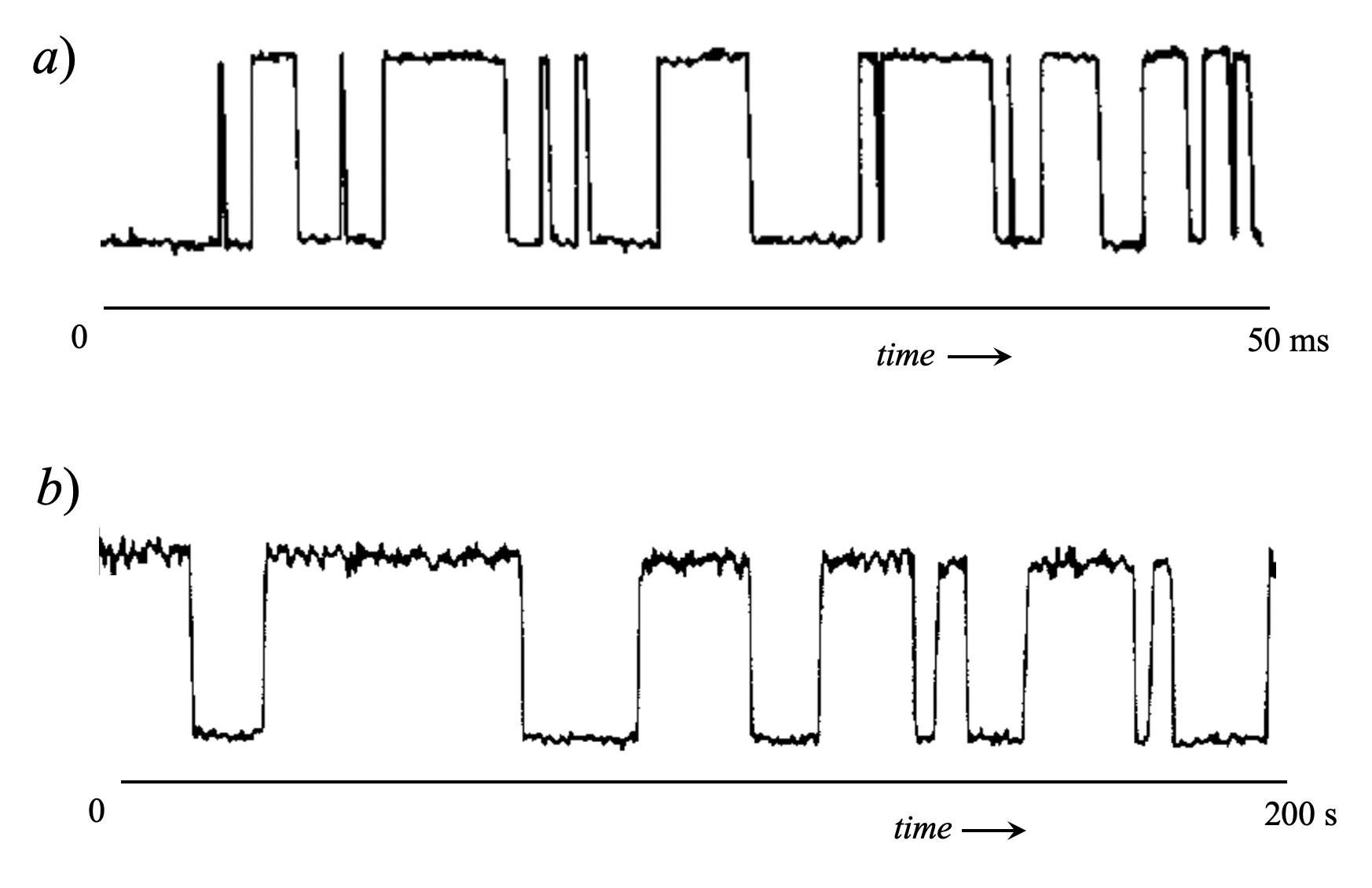}
		\end{center}
	\end{center}
	\caption{a) Data for the state of a single ion under constant illumination, showing telegraph behavior.
		Reprinted from~\cite{SauterBlattNeuhauserToschekQuantumJumps1988}.
		b) Data for the frequency of laser emission as a function of time when there are two coupled modes.   Reprinted from \cite{MorkTromborgMechanismMode1990}.}
	\label{QMexpt}
\end{figure}

Figure~\ref{QMexpt}(b) shows the temporal behavior of a classical system in which a laser hops between two metastable modes. As seen in this figure, a telegraph behavior is seen which is similar to that of the quantum single-ion behavior. A similar system was considered in Ref.~\cite{JordanSukhorukovTransportStatistics2004}.

It is of interest to ask whether these systems merely have a superficial similarity of appearance, or if there are underlying common properties.
Both of these systems are characterized by having both deterministic forces and stochastic noise. However, the noise plays a different role in the two cases. In the case of the classical telegraph switching, the system can be modeled as having two local minima in a potential energy profile~\cite{MorkTromborgMechanismMode1990}
so that the deterministic force drives the system to sit stably in one of these two minima. The noise in this case has the role of destabilizing the behavior, when a large excursion due to noise takes the system out of one of the stable minima.

By contrast, in the quantum mechanical two-state model here, the deterministic ``force'' has just one stable point, given by (\ref{avgstat}). In this case, the noise has the role of stabilizing the system at one of the two end points {\em away} from this minimum, similar to how noise can lead to metastable behavior in an inverted pendulum. An example of such process is geometric Brownian motion \cite{OksendalStochasticDifferential2003}.

Nevertheless, there is a basic similarity in that both can be described as ``truncated chaotic swings.'' Each system naturally would have large excursions, but has two ``traps'' that catch the state of the system for lengthy periods of time. The stochastic noise in each case means that the system will not stay in one of these traps forever, and will eventually pop out, to then be trapped again.


	\label{QMRes}

\section{Dependence of the Long-Term Behavior of $U_3$ on $\alpha$}
\label{sec:alpha-dependence}

Figures \ref{fig:one-path-unfiltered} and \ref{fig:varying-alpha} show that depending on the strength of the quantum noise parameter $\alpha$, there can be three qualitatively different behaviors. Large $\alpha$ leads to telegraph behavior, intermediate values of $\alpha$ cause noisy trajectories, while small values of $\alpha$ cause the process to behave mostly deterministically. Since we can only distinguish these behaviors in the long run (after $t \approx GT$, see Appendix \ref{sect.stat}), what is important is the behavior of $U_3$ in the $t \to \infty$ limit. Mathematically, that is what the stationary probability density $\rho(y)$ describes. This density would satisfy $\displaystyle \mathbb{P}(a \leq U_3 \leq b) = \int_a^b \rho(y) \,dy$ for all large times $t$.

We used the Fokker-Planck equation (see Ref.~\cite{capasso-bakstein}, Section 4.4) to derive the density $\rho(y)$ for our process $U_3$. (More details can be found in Appendix \ref{sec:fokker-planck-derivation}.) We found that
\begin{equation}\label{invariant-density-explicit}
	\rho(y) = \frac{N(\alpha, G,T)}{(1-y^2)^{2}}\exp\left\{\frac{2}{\alpha^2}\int_0^y\,\frac{-(1+q)/T + G(1-q)}{(1-q^2)^2}\;dq\right\},
\end{equation}
\noindent where $N(\alpha, G,T)$ is a normalization constant so that $\displaystyle \int_{-1}^1 \rho(y) = 1$.  Plots of $\rho(y)$ for various values of $\alpha$ were done using \textit{Mathematica} \cite{Mathematica} and are shown in Figure~\ref{Fokker}. The stationary distribution profiles agree with the results of numerical simulations in Figures \ref{fig:one-path-unfiltered} and \ref{fig:varying-alpha}.

\begin{figure}[h!]
	\begin{center}
		\begin{center}
			\includegraphics[width=0.85\textwidth]{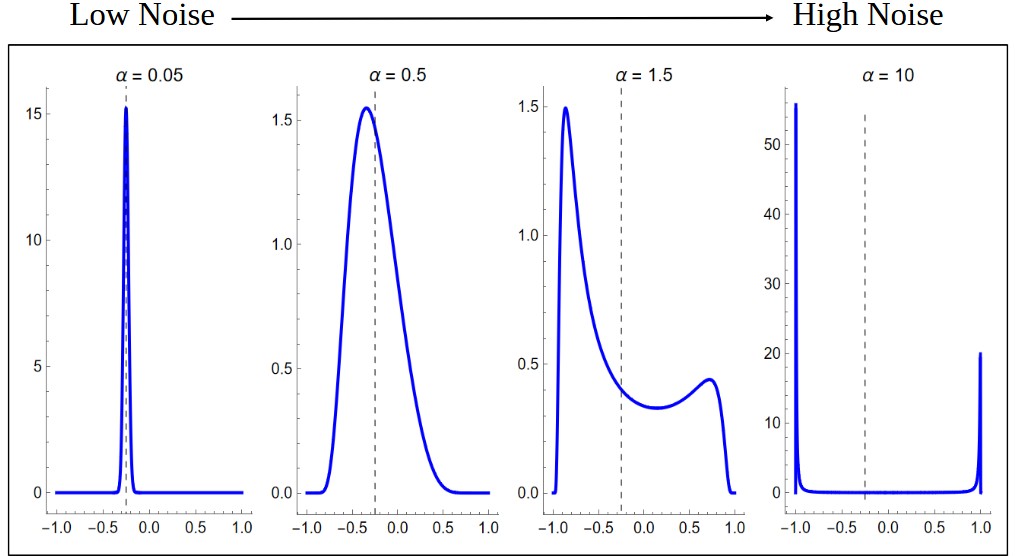}
		\end{center}
	\end{center}
	\caption{Probability distribution $\rho(y)$ of the system with $T = 1$, $G= 0.6$, and $t\rightarrow \infty$, computed using the Fokker-Planck equation, for four values of $\alpha$, the strength of the quantum noise term. The dashed vertical line represents the mean $(GT-1)/(GT+1)$.}
	\label{Fokker}
\end{figure}

\newpage

\section{Conclusions}

We have shown that a qubit subjected to constant driving and decay and continuous measurement will exhibit ``quantum jumps'' that appear as photon emission and absorption, but the mathematics involved requires no concept of a photon as a corpuscle; the behavior is the result of a noisy random walk in the presence of a deterministic driving force (Eq. \eqref{eq:SDE}).

The key feature of this model that is different from typical switching
process is that it does not require a Poisson process to initiate the ``jumps''.
The trajectories of the model in this work are completely continuous (as opposed to discontinuous trajectories coming from a Poisson process).
It also differs from the classical model that exhibits telegraph switching behavior
in a fundamental way: the stochastic noise here stabilizes the process at the upper and lower states while the stochastic noise in the classical model destabilizes the system, creating the ``jumps''.

It is, of course, possible to insist on the photon absorption/emission picture, and treat the quantum noise term (\ref{Ustoch}) as simply giving the knowledge due to weak continuous measurements. Nevertheless, this result shows that a model of continuous, spontaneous collapse is also viable in giving results that agree with experiments.

	{\bf Acknowledgements}. We thank the Pittsburgh Quantum Institute for a visiting scientist travel grant, and the Department of Mathematical Sciences at Carnegie Mellon University for the use of a computer cluster. We thank Andrew Jordan, Gautam Iyer and Bob Pego for helpful discussions.

\appendix

\section{Derivation of the driven-dissipative Bloch equations}
\label{blochder}

We start with a generic two-level system with state $|\psi\rangle = \alpha| g\rangle + \beta |e\rangle$. The density matrix for this system is 
\begin{equation}
\hat{\rho} = \left(
\begin{array}{cc}
|\alpha|^2 & \alpha^* \beta \\
\beta^*\alpha & |\beta|^2
\end{array}
\right),
\end{equation}
In the standard perturbation approach \cite{snokeQM}, we write the Hamiltonian $H = H_0 + \hat{V}$, with the first-order evolution given by
\begin{eqnarray}
\frac{d}{dt}  \langle \hat\rho_{i,j}\rangle &=&  \frac{i}{\hbar} \displaystyle \sum_n |\alpha_n|^2 \langle [ H_0+  \hat{V} ,\hat\rho_{i,j}]\rangle
\label{firstorder}
\end{eqnarray}
and the second-order evolution given by the Lindbladian,
\begin{eqnarray}
\frac{d}{dt} \langle \hat\rho_{i,j}\rangle = \frac{2\pi}{\hbar}\langle  \bigl(\hat{V}\hat\rho_{i,j}\hat{V} -\frac{1}{2} \hat{V}\hat{V}\hat\rho_{i,j} -\frac{1}{2} \hat\rho_{i,j} \hat{V}\hat{V}\bigr)\rangle \ \delta(\tilde{E}),  
\label{lindrho}
\end{eqnarray}
where the final $\delta$-function gives the energy conservation for any initial and final state.

The Hamilton for the two-state system coupled to a driving field is
 \begin{eqnarray}
H_0 =  \left(
\begin{array}{cc}
 \hbar \omega_0 & \displaystyle \frac{\hbar\omega_R}{2}e^{-i\omega t} \\
\displaystyle\frac{\hbar\omega_R}{2}e^{i\omega t} & 0
\end{array}
\right),
\label{BlochH}
 \end{eqnarray}
while the coupling to the environment is given by
\begin{eqnarray}
\hat{V} = \sum_q A_q (a_q^{ } b^\dagger_e b_g + a_q^\dagger b^\dagger_g b_e),
\end{eqnarray}
where $b^\dagger_i$, $b^{ }_i$ are the fermionic creation and destruction operators for state $i$, $a^\dagger_q$, $a^{ }_q$ are bosonic creation and destruction operators for a photon in state $q$, and
$A_q$ is a constant giving the strength of the interaction to emit or absorb a photon. This term vanishes in the first-order evolution term (\ref{firstorder}), since no phase coherence of the photons in the environment is assumed, apart from the driving laser, already accounted for in (\ref{BlochH}).


For spontaneous emission, we assume no photons initially.
Writing $\hat{\rho}_{ij} = \langle b^\dagger_i b^{ }_j\rangle$,
 and using $N_e + N_g =1$ and $(b_i)^2 = (b^\dagger_i)^2 = 0$, we obtain
\begin{eqnarray}
\label{Lindterms}
\\
&&\langle 2(a^{ }_q b^\dagger_e b^{ }_g) 
b^\dagger_e b^{ }_e (a^\dagger_q b^\dagger_g b^{ }_e) 
- (a_q b_e^\dagger b^{ }_g)(a^\dagger_q b_g^\dagger b^{ }_e)b^\dagger_e b^{ }_e - b^\dagger_e b^{ }_e
(a_q b_e^\dagger b^{ }_g)(a^\dagger_q b_g^\dagger b^{ }_e) \rangle = -2\langle N_e\rangle \nonumber \\
&&\langle 2(a^{ }_q b^\dagger_e b^{ }_g) 
b^\dagger_g b^{ }_g (a^\dagger_q b^\dagger_g b^{ }_e) 
- (a_q b_e^\dagger b^{ }_g)(a^\dagger_q b_g^\dagger b^{ }_e)b^\dagger_g b^{ }_g 
- b^\dagger_g b^{ }_g
(a_q b_e^\dagger b^{ }_g)(a^\dagger_q b_g^\dagger b^{ }_e)\rangle = 2\langle N_e\rangle \nonumber \\
&&\langle 2(a^{ }_q b^\dagger_e b^{ }_g) 
b^\dagger_e b^{ }_g (a^\dagger_q b^\dagger_g b^{ }_e) 
- (a_q b_e^\dagger b^{ }_g)(a^\dagger_q b_g^\dagger b^{ }_e)b^\dagger_e b^{ }_g 
- b^\dagger_e b^{ }_g
(a_q b_e^\dagger b^{ }_g)(a^\dagger_q b_g^\dagger b^{ }_e)\rangle =  -\langle b_e^\dagger b^{ }_g \rangle
\nonumber \\
\nonumber
&&\langle 2(a^{ }_q b^\dagger_e b^{ }_g) b^\dagger_g b^{ }_e (a^\dagger_q b^\dagger_g b^{ }_e) 
- (a_q b_e^\dagger b^{ }_g)(a^\dagger_q b_g^\dagger b^{ }_e)b^\dagger_g b^{ }_e 
- b^\dagger_g b^{ }_e(a_q b_e^\dagger b^{ }_g)(a^\dagger_q b_g^\dagger b^{ }_e)\rangle = - \langle b_g^\dagger b^{ }_e \rangle.
\nonumber
\end{eqnarray}
Defining
\begin{eqnarray}
U_1 &=& \beta^*\alpha + \alpha^*\beta  \nonumber\\
U_2 &=& i(\beta^*\alpha- \alpha^*\beta) \nonumber\\
U_3 &=& |\beta|^2 -|\alpha|^2,
\label{blochcomp}
\end{eqnarray}
and $|\alpha|^2 + |\beta|^2 = 1$,
we then have for the Lindbladian term
\begin{eqnarray}
\frac{\partial U_3}{\partial t} &=& - \frac{1}{T_1}2|\beta|^2 =  - \frac{1}{T_1}(U_3+1),
\end{eqnarray}
where $T$ is the characteristic photon emission time, and
\begin{eqnarray}
\frac{\partial U_1}{\partial t} &=&  - \frac{1}{2T_1}U_1, \nonumber\\
\frac{\partial U_2}{\partial t} &=&  - \frac{1}{2T_1}U_2.
\end{eqnarray}
These three equations give the standard dissipation terms in the Bloch equations. 

We can now adopt the rotating frame and writing
$|\psi\rangle = \alpha'|v\rangle+\beta' e^{-i\omega_0t}|c\rangle$. Writing the Bloch vector in terms of $\alpha'$ and $\beta'$, for the case when the coherent driving field frequency $\omega$ equals $\omega_0$, Schr\"odinger's equation for the Hamiltonian (\ref{BlochH}) plus the dissipative terms give the standard equations in the rotating frame,
\begin{eqnarray}
\displaystyle \frac{\partial U'_1}{\partial t} &=& \displaystyle - \frac{U'_1}{T_2} \nonumber\\
\displaystyle \frac{\partial U'_2}{\partial t} &=& \displaystyle - \frac{U'_2}{T_2} - \omega_R U'_3
\nonumber\\
\displaystyle \frac{\partial U'_3}{\partial t} &=& \displaystyle - \frac{U'_3+1}{T_1} +\omega_R U'_2 ,
\label{Blocheq}
\end{eqnarray}
with $T_2 = 2T_1$. This assumes that the
the photon emission gives both relaxation from the upper state and also decoherence; it is also possible to have additional decoherence processes that do not lead to relaxation, giving an effective $T_2$ that is shorter.

We can now take exactly the same approach to account for incoherent pumping. In this case, analogous to (\ref{Lindterms}), we start with an initial state with one photon in state $q$, and write
\begin{eqnarray}
\\
&&\langle 2(a^\dagger_q b^\dagger_g b^{ }_e) 
b^\dagger_e b^{ }_e (a_q b^\dagger_e b^{ }_g) 
- (a^\dagger_q b_g^\dagger b^{ }_e)(a_q b_e^\dagger b^{ }_g)b^\dagger_e b^{ }_e 
- b^\dagger_e b^{ }_e
(a^\dagger_q b_g^\dagger b^{ }_e)(a_q b_e^\dagger b^{ }_g) \rangle = 2\langle N_g\rangle \nonumber \\
&&\langle 2(a^\dagger_q b^\dagger_g b^{ }_e) 
b^\dagger_g b^{ }_g (a_q b^\dagger_e b^{ }_g) 
- (a^\dagger_q b_g^\dagger b^{ }_e)(a_q b_e^\dagger b^{ }_g)b^\dagger_g b^{ }_g 
- b^\dagger_g b^{ }_g
(a^\dagger_q b_g^\dagger b^{ }_e)(a_q b_e^\dagger b^{ }_g) \rangle = - 2\langle N_g\rangle . \nonumber 
\end{eqnarray}
These give us, for the same definitions as above,
\begin{eqnarray}
\frac{\partial U_3}{\partial t} &=&  2G|\alpha|^2 =   G(1-U_3),
\end{eqnarray}
which is the incoherent pumping term used in the main text. Similarly, the off-diagonal terms are
\begin{eqnarray}
\\
&&\langle 2(a^\dagger_q b^\dagger_g b^{ }_e) 
b^\dagger_e b^{ }_g (a_q b^\dagger_e b^{ }_g) 
- (a^\dagger_q b_g^\dagger b^{ }_e)(a_q b_e^\dagger b^{ }_g)b^\dagger_e b^{ }_g 
- b^\dagger_e b^{ }_g
(a^\dagger_q b_g^\dagger b^{ }_e)(a_q b_e^\dagger b^{ }_g)\rangle = - \langle b_e^\dagger b^{ }_g \rangle,
\nonumber \\
&&\langle 2(a^\dagger_q b^\dagger_g b^{ }_e) 
b^\dagger_g b^{ }_e (a_q b^\dagger_e b^{ }_g) 
- (a^\dagger_q b_g^\dagger b^{ }_e)(a_q b_e^\dagger b^{ }_g)b^\dagger_g b^{ }_e 
- b^\dagger_g b^{ }_e
(a^\dagger_q b_g^\dagger b^{ }_e)(a_q b_e^\dagger b^{ }_g)\rangle =  -\langle b_g^\dagger b^{ }_e \rangle
\nonumber 
\end{eqnarray}
which give decoherence due to the incoherent pumping. This  can then be added to the first two equations of (\ref{Blochstandard}) to give an adjusted $T_2$ time.





Note that the Lindbladian approach used here gives {\em irreversible} behavior, namely deterministic decay, but cannot give {\em stochastic} behavior, that is, different results for the same initial conditions. This is fundamentally because any approach derived from unitary quantum mechanics, including the Lindbladian term used here, acts only on the many-body wave function of the system plus its enviroment, and therefore gives deterministic behavior. This is fundamentally the same result as found in Ref.~\onlinecite{SnokeLiuGirvinBasisSecond2012}.

\section{Justification of the Statistics}
\label{sect.stat}

In this section, we provide mathematical justification for the fact that $\displaystyle \langle t_e \rangle/\langle t_g\rangle \approx G T$  when $\alpha$ is a fixed large number, as shown in Figure \ref{fig:ratio}. (In Section \ref{sec:alpha-dependence}, we justify the qualitatively different behavior of trajectories in Figures \ref{fig:one-path-unfiltered} and \ref{fig:varying-alpha} as a function of $\alpha$.) Starting from \eqref{eq:SDE}, let us consider the shifted process $V = U_3 + 1$. The process $V$ satisfies

\begin{equation}
	\label{eq:SDE2}
	d V = \left(-\frac{V}{T} + G(2 -V) \right)  dt + \alpha V(2-V) dW_t \,.
\end{equation}

For large $\alpha$, Figure \eqref{fig:one-path-unfiltered} (a) shows that the process $U_3$ exhibits telegraph behavior. We can therefore neglect the time spent in transition between the metastable states and estimate the ratio $\displaystyle \langle t_e \rangle/\langle t_g\rangle$ by the average amount of time any qubit starting in the upper state takes to decay out of that state divided by the average time it takes for any qubit starting in the lower state to become excited out of that state. More precisely, if $0 < \delta \ll 1$, we can define the upper state as the interval $[2-\delta, 2]$ and the lower state as the interval $[0,\delta]$. (In Figure \ref{fig:ratio}, we took $\delta = 0.05$.) In the large $\alpha$ limit, for $V_0 = 2$ we therefore should have $\displaystyle \langle t_e \rangle = \mathbb{E}[\tau | V_0 = 2]$, and $\displaystyle \langle t_g \rangle = \mathbb{E}[\tau | V_0 = 0]$, where $\displaystyle \tau = \inf_{t > 0}\{\delta < V_t < 2-\delta\}$, and $\mathbb{E}$ denotes the conditional expectation.

Computing the exit time $\tau$ may be difficult to do directly for the general case of \eqref{eq:SDE2}. However, if $V \approx 0$, then the drift term in $\eqref{eq:SDE2}$ is well-approximated  by the constant drift $2G$, and the noise term is well-approximated by the linear term proportional to $2\alpha V$ (see Appendix \ref{sec:linear-justification}). Therefore, for $V \approx 0$, we have that $X_t \approx V_t$, where $X_t$ satisfies:

\begin{equation}
	\label{eq:SDE2-approx-lower}
	d X = 2 G  dt +2 \alpha X dW_t \,.
\end{equation}

The exit time of $X_t$ is $\displaystyle \tau' = \inf_{t > 0}\{\delta < X_t < 2-\delta\}$. We can use Dynkin's formula (see Ref. \cite{capasso-bakstein}, p. 195, formulas (4.56)-(4.57)) to compute this approximate exit time. In the case $X_0 = 0$, we need to solve the second-order ODE

\begin{equation}\label{eq:ODE-approx-lower}
	\displaystyle 2G\frac{df}{dx} + \frac{(\alpha x)^2}{2}\frac{d^2f}{dx^2} = -1,\; 0 \leq x < \delta, \quad  f(\delta) = 0.
\end{equation}

\noindent The only \textit{bounded} solutions to the ODE in \eqref{eq:ODE-approx-lower} on the interval $x \in [0,\delta]$ are linear functions of the form $f(x) = -x/(2G) + c_1$. Enforcing the boundary condition  $f(\delta) = 0$ yields $f(x) = (\delta -x )/(2G)$. We therefore have that $\tau' = f(0) = \delta /(2G)$. We deduce that $\langle t_g \rangle \approx \delta/(2G)$.

To find $\langle t_e \rangle$, we can follow the same approach for $V \approx 2$. In this case, the drift term in $\eqref{eq:SDE2}$ is well-approximated  by the constant drift $-2/T$, and the noise term is approximately proportional to $-2\alpha(V-2)$ (see Appendix \ref{sec:linear-justification}). Therefore, for $V \approx 2$, we have that $X_t \approx V_t$, where $X_t$ satisfies:

\begin{equation}
	\label{eq:SDE2-approx-upper}
	d X = -(2/T)  dt -2 \alpha(X-2) dW_t \,.
\end{equation}

In this case, $X_0 = 2$, and we need to solve this second-order ODE:

\begin{equation}\label{eq:ODE-approx-upper}
	\displaystyle -(2/T)\frac{df}{dx} + 2\alpha^2(x-2)^2\frac{d^2f}{dx^2} = -1,\; 2-\delta < x \leq 2, \quad  f(2-\delta) = 0.
\end{equation}

\noindent The bounded solution on the interval $x \in [2-\delta, 2]$ is $f(x) = T(x - 2 + \delta)/2$. We therefore have that $\tau' = f(2) = T\delta /2$. In other words,  $\langle t_e \rangle \approx T\delta/2$.

Finally, we now see that $\displaystyle \frac{\langle t_e \rangle}{\langle t_g\rangle} \approx \frac{T\delta/2}{ \delta/(2G)} = G T$  when $\alpha$ is large.

\section{A justification for the linear approximation near the endpoints}
\label{sec:linear-justification}

In this section, we give an estimate of the difference between the between the exact solution of~\eqref{eq:SDE2}
and the approximate solution of \eqref{eq:SDE2-approx-lower} when the initial data is in $(0,\delta)$, before
exiting this region.

\begin{lemma}
	\label{lem:bounded-exit}
	Suppose $\delta \ll GT$ be such that $2GT > \delta(1 + GT)$.
	Let $V$ be the solution of~\eqref{eq:SDE2}  with initial data $V_0 = x \in (0,\delta)$ and
	$\tau = \inf \left\{ t>0: V_t = \delta  \right\}$.
	Then, there exists a constant $C$ such that
	\begin{equation}
		\sup_{x\in (0,\delta)} \E^x[\tau] \leq \delta/G\,.
		\label{ine:exit-time}
	\end{equation}
\end{lemma}

{
\begin{proof}
	Let $f = \E^x \tau$. Then $f$ is the bounded solution of the following equation
	\begin{equation}
		\label{eq:ODE-full}
		- \frac{1}{T}\left( - x + GT (2 - x) \right) f'
		- \frac12\alpha^2 x^2(2-x)^2 f'' = 1
	\end{equation}
	with boundary condition $f(\delta) = 0$.
	We note that the endpoint $x = 0$ is inaccessible so
	one does not specify boundary condition there~\cite{feller_parabolic_1952, feller_diffusion_1954}.
	Let $g(x) = ( \delta - x )/G$.
	Note that for $x \in (0,\delta)$, $g(x) \leq \delta/G$
	is a supersolution of~\eqref{eq:ODE-full}  with boundary $g(\delta) = 0$.
	To see this, we note that for $\delta$ satisfying the assumption,
	\begin{align*}
		 & -\frac{1}{T} ( - x + GT(2-x))\left(\frac{-1}{G} \right)            \\
		 & \geq \frac{1}{GT}\left( -\delta(1 + GT) + 2GT   \right) \geq 1 \,.
	\end{align*}
	By comparison principle from the theory of partial differential equations~\cite{evans_partial_2010}, $g \geq f$
	and estimate~\eqref{ine:exit-time} then follows as $g \leq \delta/G$.
\end{proof}

The main result of this section is the following proposition.
\begin{proposition}
	Let $V$ be the solution of~\eqref{eq:SDE2} and $X$ solution of~\eqref{eq:SDE2-approx-lower} with initial data $V_0 = X_0 = x \in (0,\delta)$.
	Let $Y = V-X$.
	Then, there exists a constant $C>0$ such that
	\begin{equation}
		\E Y_{t\wedge \tau}^2 \leq
		C\delta^3\left( \left(\frac1T + G \right)^2 + \frac12\alpha^2 \delta^4\right)  e^{12t} \,.
		\label{ine:endpoints}
	\end{equation}
\end{proposition}

\begin{proof}
	Subtracting~\eqref{eq:SDE2-approx-lower} from~\eqref{eq:SDE2} , we have,
	\begin{align*}
		dY
		 & = \left( -\frac{V_t}{T} + G(2-V_t) - 2G \right) dt + \alpha \left( 2Y_t -V_t^2\right) \, dW_t     \\
		 & = \left( -\frac{V_t}{T} -GV_t \right) dt + \alpha \left( 2Y_t -V_t^2\right) \, dW_t           \,.
	\end{align*}
	For $V_0 = X_0 = x \in (0,\delta)$ and $t \leq \tau$, we know that $0 \leq V \leq \delta$.
	Therefore, by It\^o formula,
	\begin{equation*}
		dY_t^2 = 2Y_t \, dY_t + \alpha^2 ( 2Y_t - V_t^2)^2 \, dt
	\end{equation*}
	and by Cauchy-Schwarz inequality,
	\begin{align*}
		\E Y_{t\wedge \tau}^2
		 & = -2 \E\int_0^{t \wedge \tau} Y_s\left( \frac{V_s}{T} + GV_s   \right) \, ds
		+ \alpha^2 \E \int_0^{t \wedge \tau} ( 2Y_s - V_s^2)^2 \, ds                                           \\
		 & \leq 4 \E \int_0^{t \wedge \tau} \left( Y_s^2 + V_s^2\left( \frac{1}{T} + G \right)^2 \right) \, dt
		+ 2\alpha^2 \E \int_0^{t \wedge \tau} \left( 4Y_s^2 + V_s^4 \right) \, ds                              \\
		 & \leq 4\delta^2\left( \left(\frac1T + G \right)^2 + \frac12\alpha^2 \delta^4\right) \E [\tau]
		+ 12\E \int_0^{t} Y_{s\wedge \tau}^2 \, ds \,.
	\end{align*}
	By Gronwall's inequality, we have
	\begin{equation*}
		\E Y_{t\wedge \tau}^2 \leq
		4\delta^2\left( \left(\frac1T + G \right)^2 + \frac12\alpha^2 \delta^4\right) \E [\tau] e^{12t} \,.
	\end{equation*}
	Inequality~\eqref{ine:endpoints} then follows by Lemma~\ref{lem:bounded-exit}.
\end{proof}
}
When $t \approx \E[\tau]$ for example, the difference between $V$ and $X$ is of order $O(\delta^3)$, 
which means that $X$ is a good approximation of $V$.

\section{Details about the Fokker-Planck Equation}
\label{sec:fokker-planck-derivation}

In general, if a stochastic process $X_t$ admits a time-dependent density $\rho(s,x;t,y)$, then  whenever $t > s$, the density satisfies:
\begin{equation}
	\mathbb{P}(a \leq X_t \leq b | X_s = x) = \int_a^b \rho(s,x;t,y) \; dy.
\end{equation}

\noindent Here, $\mathbb{P}(a \leq X_t \leq b | X_s = x)$ is the conditional probability that $ a \leq X_t \leq b$, given initial condition $X_s = x$. It is well-known that  for an It\^{o} process $dX_t = u(t,X_t)dt + \sigma(t,X_t)dW_t$ the density satisfies the Fokker-Planck equation (see Ref.~\cite{capasso-bakstein}, Section 4.4). This is a partial differential equation for $\rho$ as a function of $t$ and $y$, keeping $s$ and $x$ fixed:
\begin{equation}\label{eq:Fokker-Planck}
	\frac{\partial \rho}{\partial t}(s,x;t,y) =  -\frac{\partial}{\partial y}[u(t,y)\rho(s,x;t,y)] + \frac{1}{2}\frac{\partial^2}{\partial y^2}[\sigma^2(t,y)\rho(s,x;t,y)].
\end{equation}

For our process $U_3$,  $u(t,y) = u(y) = -(1+y)/T + G(1-y)$ and $\sigma(t,y) = \sigma(y)= \alpha(1-y^2)$. Note that $U_3$ is autonomous. We also require that $\displaystyle \frac{\partial \rho}{\partial t} \equiv 0$, since we seek a stationary solution. Hence, $\rho$ will be a function of $y$ only, and the Fokker-Planck equation reduces to this ODE:
\begin{equation}\label{eq:Fokker-Planck-stationary}
	0 =  -\frac{d}{d y}[u(y)\rho(y)] + \frac{1}{2}\frac{d^2}{d y^2}[\sigma^2(y)\rho(y)].
\end{equation}
This can be solved by hand, but we simply refer to Ref.~\cite{risken1996fokker}, p. 98, (5.13), for the solution:
\begin{equation}\label{invariant-density}
	\rho(y) = \frac{N(\alpha,G,T)}{\sigma(y)^2}\exp\left\{\int_0^y\frac{2 u(q)}{\sigma(q)^2}\;dq\right\},
\end{equation}
\noindent where $N(\alpha, G,T)$ is a normalization constant so that $\displaystyle \int_{-1}^1 \rho(y) = 1$. (Even though $\sigma(1) = \sigma(-1) = 0$, $\rho$ is in fact normalizable. This is because the limits $\displaystyle\lim_{y\to 1^-}\rho(y)$ and $\displaystyle \lim_{y\to -1^+}\rho(y)$ exist and can be shown to equal $0$, so that $\rho$ is bounded.)



\vspace{0.3in}

\bibliography{bibo, theBib}

\end{document}